\documentclass[journal]{IEEEtran}
\ifCLASSINFOpdf
\else
\fi
\usepackage{amssymb}
\usepackage{amsmath}
\usepackage{amsthm}
\usepackage{algpseudocode}
\usepackage{graphicx}
\usepackage{booktabs}
\usepackage{cite}
\usepackage{bm}
\usepackage{float}
\usepackage{multirow}
\usepackage{color}
\usepackage{subfigure}
\usepackage{caption}
\usepackage{epstopdf}
\usepackage{hyperref}
\usepackage{fixltx2e}
\usepackage{longtable}
\usepackage{diagbox}
\usepackage{changepage}
\usepackage{comment}
\usepackage{cases}
\usepackage[short]{newoptidef}
\usepackage{makecell}
\usepackage{mathabx}
\theoremstyle{definition}

\newtheorem{definition}{Definition}
\usepackage{algorithm}
\theoremstyle{remark}

\theoremstyle{plain}
\newtheorem{theorem}{Theorem}
\newtheorem{lemma}{Lemma}

\newtheorem{corollary}{Corollary}

\begin{document}
\title{Serving Long-Context LLMs at the Mobile Edge: Test-Time Reinforcement Learning-based Model Caching and Inference Offloading}

\author{Minrui Xu, Dusit Niyato, \emph{Fellow, IEEE}, Christopher G. Brinton, \emph{Senior Member, IEEE}
\thanks{M.~Xu and D.~Niyato are with the College of Computing and Data Science, Nanyang Technological University, Singapore 639798, Singapore (e-mail: minrui001@e.ntu.edu.sg; dniyato@ntu.edu.sg). Christopher G. Brinton is with the Elmore Family School of Electrical and Computer Engineering, Purdue University, West
Lafayette, IN 47906 USA (e-mail: cgb@purdue.edu).}
}

\maketitle

\begin{abstract}
% 第一句： 大语言模型在边缘计算的大背景
% 第二句：Serving LLMs 需要解决CAP问题
% 第三句：解决CAP问题是困难的
% 第四句：我们目前
Large Language Models (LLMs) can perform zero-shot learning on unseen tasks and few-shot learning on complex reasoning tasks. However, resource-limited mobile edge networks struggle to support long-context LLM serving for LLM agents during multi-round interactions with users. Unlike stateless computation offloading and static service offloading in edge computing, optimizing LLM serving at edge servers is challenging because LLMs continuously learn from context which raises accuracy, latency, and resource consumption dynamics. 
In this paper, we propose a joint model caching and inference offloading framework that utilizes test-time deep reinforcement learning (T2DRL) to optimize deployment and execution strategies for long-context LLM serving. In this framework, we analyze the performance convergence and design an optimization problem considering the utilization of context windows in LLMs. Furthermore, the T2DRL algorithm can learn in both the training phase and the testing phase to proactively manage cached models and service requests and adapt to context changes and usage patterns during execution. To further enhance resource allocation efficiency, we propose a double Dutch auction (DDA) mechanism, which dynamically matches supply and demand while maximizing social welfare. Finally, experimental results demonstrate that the T2DRL algorithm can reduce system costs by at least 30\% compared to baselines while guaranteeing the performance of LLM agents in real-world perception and reasoning tasks.
% Previous frameworks for computation and service offloading are inadequate for optimizing LLM services as they fail to capture the dynamic learning capabilities of LLMs and the need for service adaptation based on evolving model contexts. 
\end{abstract}


% \fi
\IEEEpeerreviewmaketitle

\section{Introduction}

Large Language Models (LLMs)~\cite{xu2024large} have billions of parameters and demonstrate emerging capability through pre-training on internet-scale datasets. In mobile edge networks, LLMs can empower mobile devices to precept environments and plan for future actions, for running LLM agents on devices and enabling zero-shot learning for unseen perception 
tasks and few-shot learning for complex reasoning tasks~\cite{brown2020language, bubeck2023sparks}. For instance, LLMs can empower mobile devices with environmental understanding and reasoning capabilities, allowing them to adapt flexibly to changing conditions~\cite{zhou2024large}. Furthermore, LLMs improve the experience of user interaction and achieve automation for mobile devices during the inference process by leveraging the contextual data derived from environmental perceptions and user interactions. However, how to achieve a balance among context, accuracy, and performance for long-context LLM serving in resource-constraint mobile edge networks is still an open question~\cite{zeng2024cap}.

\begin{figure}[t]
    \centering
    \includegraphics[width=0.9\linewidth]{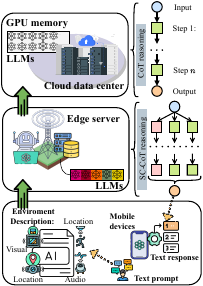}
    \caption{Serving LLMs to handle inputs and tackle complex tasks using CoT prompting in mobile edge networks.}
    \label{fig:system}
\end{figure}

As illustrated in Fig.~\ref{fig:system}, resource-intensive tasks such as LLM serving, i.e., the deployment and inference of LLMs to perform inference tasks, in mobile edge networks~\cite{zeng2024cap}, often need to be offloaded to edge servers (ESs) or cloud data centers to reduce latency and improve resource utilization efficiency~\cite{ye2024galaxy}. Due to the unique context-aware characteristics of LLMs that they can learn in-context examples during few-shot learning~\cite{brown2020language}, LLM serving at ESs is significantly different from computation and service offloading~\cite{pan2022retention, liu2023cache}. Specifically, computation offloading involves transferring computation-intensive tasks from local devices to ESs or a cloud data center~\cite{zhao2022edgeadaptor}, which is stateless and independent of previous computation outcomes or the historical results of the servers. Similarly, service offloading involves sending user inputs to ESs to invoke pre-cached services~\cite{ma2020cooperative}, which is static and remains unchanged across different invocations. As a result, existing optimization frameworks for computation offloading and service offloading are no longer effective for optimizing long-context LLM serving, as they fail to accommodate the dynamic learning capabilities of LLMs and the need for service adaptation based on evolving model contexts.

Similar to content caching and service caching in mobile edge networks, model caching for LLMs emerges as a promising solution for provisioning LLM agents with cached LLMs at ESs~\cite{hu2024memserve, singh2024llm}. Specifically, LLMs can learn from input examples and adjust the output by eliciting their reasoning capability via prompting techniques, such as Chain-of-Thought (CoT)~\cite{wei2022chain} and Self-Consistent CoT~\cite{wang2022self} prompting. These CoT prompting can be utilized to improve accuracy but introduce additional resource consumption and latency. Nevertheless, the performance of the LLM agent is not only constrained by the hardware resources available on ESs but also limited by the context windows of LLMs, which are determined by the model architecture. If the model’s inference tokens exceed the context window, the cached model needs to be evicted; otherwise, it will result in a significant decline in model performance. To optimize long-context LLM serving in mobile edge networks, it is essential to actively manage the cached LLMs and model inference requests, adapting to the state transitions of inference requests and cached models.

In this paper, we propose a joint model caching and inference offloading framework for long-context LLM serving in mobile edge networks, which allows LLMs to adopt various prompting techniques. In this framework, we formulate the problem of long-context LLM serving as a mixed-integer programming problem to minimize the system cost under the hardware resource constraints and context window constraints. To evaluate the performance of LLMs theoretically, we analyze a convergence rate of language ambiguity with respect to the reasoning paths and steps. To obtain optimal cache and inference management strategies, we propose a test-time deep reinforcement learning (T2DRL) algorithm~\cite{schulman2017proximal, sun2024learning} utilizing test-time training (TTT) models as the actor-critic network that can learn not only from historical experience during training but also from real-time interaction during testing. By treating the hidden state as a learnable model, the TTT model-based policy can be updated through self-supervised learning during test time, which overcomes limited expressive power in Long Short-Term Memory (LSTM) networks and the quadratic complexity of attention mechanisms. Based on the optimization results of T2DRL, we formulate the long-context LLM serving market and introduce a double Dutch auction (DDA) mechanism to further enhance the efficiency of resource allocation for maximizing social welfare. This auction dynamically matches the supply of edge server resources with the demand from user devices, ensuring fair and cost-effective distribution of computational resources. In this way, DDA can enhance long-context LLM serving by providing efficient resource allocation, dynamic pricing, and real-time adaptation capabilities that optimize computational resources across edge servers, particularly benefiting complex tasks like multi-path reasoning while maintaining low latency and maximizing social welfare. By leveraging buyer and seller clocks, the auction process adapts to the real-time needs of long-context LLM inference tasks, providing a scalable and transparent solution for dynamic environments. Finally, the experimental results demonstrate that the proposed algorithm can minimize system costs compared with existing baselines.

% We contribute a comprehensive solution comprising a unified framework, a theoretical analysis of its convergence properties, and an innovative offline reinforcement learning algorithm that minimizes overall system costs. Our experimental evaluations demonstrate that the proposed methods significantly reduce the costs associated with serving LLMs while maintaining high performance, thereby enabling practical and efficient LLM deployment on resource-limited ESs. The remainder of the paper is organized as follows: an exploration of related work, detailed methodology of our proposed framework, algorithmic strategies and their theoretical justification, followed by a presentation of our experimental results and concluding remarks on future research directions.

The main contributions can be summarized as follows.

\begin{itemize}
    \item For long-context LLM serving in mobile edge networks, we propose a joint model caching and inference offloading framework to determine resource allocation and deployment strategies for LLM agents in mobile edge networks, which can adapt to dynamic network conditions and device capabilities.
    \item In the framework, we analyze the convergence rate of LLMs handling contexts with CoT prompting. This approach enhances the reasoning and planning capability of LLMs by generating intermediate reasoning thoughts with additional latency and resources.
    \item Furthermore, we model a long-term optimization problem to minimize total system costs in terms of edge inference cost at ESs and the cloud inference cost of the cloud data center. We formulate this long-context LLM serving problem as a Markov decision process (MDP) where model caching and inference offloading decisions can affect the transition of states.
    \item We introduce a Double Dutch Auction mechanism to match supply and demand efficiently, ensuring fairness, transparency, and cost-effectiveness in the allocation of computational resources for serving long-context LLMs.
    \item To achieve proactive cached model and inference management, we propose a test-time deep reinforcement learning algorithm that learns optimal model caching strategies with a test-time training model-based actor-critic network. This algorithm can adjust to the dynamic requests and contexts during the testing phase, balancing performance with computational efficiency.
\end{itemize}

The remainder of this paper is organized as follows. Section II reviews related work on serving LLMs in mobile edge networks, caching strategies, and reinforcement learning approaches. Section III introduces the system model and problem formulation. Section IV details the proposed T2DRL algorithm and its optimization process. Section V discusses the market formulation and the DDA mechanism. Section VI presents experimental results and performance evaluations. Finally, Section VII concludes the paper and outlines future research directions.
% 第一段: 大背景

% 第二段: 问题

% 第三段: 挑战

% 第四段: 本文的

\section{Related Work}

\subsection{Large Language Model Serving in Mobile Edge Networks}

In mobile edge networks, there are three main objectives for serving LLMs including context, accuracy, and performance~\cite{zhao2022edgeadaptor, zeng2024cap}. Context is the ability of LLMs to handle longer contextual inputs effectively, which can improve their ability to understand and generate responses based on more extensive and potentially more complex histories during inference. The inference performance can be measured by inference speed, throughput (e.g., tokens per second), and cost efficiency (e.g., cost per token). To balance the accuracy and performance of LLM serving in 6G, Lin \textit{et al.} in \cite{lin2023pushing} propose a collaborative framework for efficient edge training and inference, leveraging advanced techniques like split learning and quantization to achieve low-latency inference despite limited edge computing resources. 
% To ensure efficient LLM serving in edge-cloud collaboration framework, Qian \textit{et al.} in \cite{qian2023user} propose the Dinkelbach algorithm, alternating optimization, semidefinite relaxation (SDR), the Hungarian method, and a fractional programming technique to reformulate optimization problems as Quadratically Constrained Quadratic Programming (QCQP), making them solvable by SDR and the Hungarian algorithm. 
Furthermore, a personalized inference scheduling framework is proposed in \cite{yang2024perllm}, which formulates LLM serving as a combinatorial multi-armed bandit problem and proposes a constraint satisfaction upper confidence bound (UCB) algorithm for inference scheduling and resource allocation. To preserve data privacy during LLM serving, Su \textit{et al.} in \cite{su2024titanic} present a novel distributed training paradigm designed specifically for fine-tuning pre-trained LLMs in a privacy-preserving manner directly on client devices. From the perspective of resource utilization efficiency, Ye \textit{et al.} in \cite{ye2024galaxy} propose Galaxy by incorporating a heterogeneity and memory-budget aware workload planning algorithm and Yuan \textit{et al.} in \cite{yuan2024generative} use the UCB algorithm to determine energy-efficient configurations.

\subsection{Service Caching and Offloading in Mobile Edge Networks}

Service caching and offloading allow for the full utilization of heterogeneous edge resource capacities for executing computation-intensive services in mobile edge networks. To minimize service provisioning delay, Zhou \textit{et al.} in \cite{zhou2022two} leverage the Lyapunov optimization approach and dependent rounding technique to determine caching and offloading decisions under convergence constraints. Considering cooperative service caching and workload scheduling, Ma \textit{et al.} in \cite{ma2020cooperative} propose an iterative algorithm based on Gibbs sampling and a heuristic workload scheduling algorithm with polynomial complexity. To enhance the efficiency of Federated Learning (FL), a framework of cache-based FL model training is developed in \cite{liu2023cache} for reducing per-iteration training time. Furthermore, edge caching for model downloading is investigated in \cite{qu2024trimcaching}. During LLM serving, the concept of context caching is proposed in \cite{hu2024memserve} that adopts a paging-based dynamic memory management system and explores efficient Key-Value (KV) cache transfer methods.
% \cite{singh2024llm}

\subsection{Deep Reinforcement Learning for Edge Caching}

DRL can learn proactive caching decisions to enhance cache hit rates by allowing the system to predict future content requests and cache popular content in dynamic edge systems~\cite{wang2022a3c}. To handle the temporal-spatial dynamic cache popularity, LSTM networks, and attention mechanisms are leveraged to improve the performance of DRL. For instance, LSTM-based DRL can dynamically adapt to changing user behaviors and environmental conditions, making them more effective than static caching strategies~\cite{narayanan2018deepcache, xie2024deep}. Moreover, attention-based DRL can capture the influence of adjacent ESs on the current node, which helps in making more informed caching decisions, which leads to improved caching efficiency and reduced access to the core network~\cite{zhao2021neighboring, yao2023cooperative}. Considering the increasing popularity and demand for LLM applications on mobile devices, He {et al.} in \cite{he2024large} propose an active inference approach for offloading and resource allocation of LLM inference tasks in the edge-cloud collaboration framework, which dynamically adapts to varying task loads and resource constraints in mobile edge networks. Nevertheless, during long-context inference, LSTM has limited expressive power in their hidden states while attention mechanisms have a quadratic complexity~\cite{sun2024learning}. Therefore, in this paper, we proposed a novel T2DRL algorithm with linear complexity and an expressive hidden state to handle long contexts efficiently.

\section{System Model}

In mobile edge networks, users can access LLM agents to answer questions or handle complex tasks autonomously for them, each LLM agent is empowered by one or multiple LLMs. Every user group engages a single LLM agent as their active assistant based on their specific needs and tasks. The network comprises $N+1$ LLM agent providers including $N$ ESs equipped with computing units and a centralized cloud data center. The set of LLM agent providers is denoted by $\mathcal{N} = \{0, 1, \ldots, N\}$ and the cloud data center by $0$. These ESs can execute LLM agents, while additional inference requests can be offloaded to the cloud data center. The set of available LLM agents is represented as $\mathcal{I}=\{1,2,\ldots, I\}$, which rely on one or multiple LLMs whose set is denoted by $\mathcal{M}=\{1, \ldots, M\}$. Given the multi-functionality of LLMs in providing various LLM agents, it is typically observed that the number of LLM agents, $I$, is significantly more than the number of LLMs, i.e., $I \gg M$. For each user group $\mathcal{U}_n$ associated with ES $n$, the demand for LLM agents is represented by $R_{n}^t = \{R_{n, i, m}^t | i \in \mathcal{I}, m \in \mathcal{M}\}$, which presents the number of requests for LLM $m$ to perform specific tasks. The input data size for LLM agent $i$ is denoted by $d_i$, and the configuration of each LLM $m$ includes model size $s_m$, computation required per token $e_m$, and context window size $w_m$. The context window size is a unique characteristic of LLMs, which refers to the maximum number of tokens that an LLM can effectively process and utilize in long-term inference tasks.

\subsection{Network Model}

% To facilitate interaction with LLM agents, each user group covered by a specific LLM service provider, denoted by $\mathcal{U}_n$, can access services directly through their respective ES $n$. This arrangement leverages the existing terrestrial communication infrastructures to manage data transmission. Users within the same LLM service provider share spectral resources, leading to potential mutual interference among them. The channel power gain from mobile user $u \in \mathcal{U}_n$ to their ES $n$ incorporates factors such as large-scale fading and Rayleigh fading, typical of terrestrial wireless communications, and is denoted by $g_{u, n}$. The bandwidth allocated to each ES $n$ is represented as $B_n$. Consequently, the uplink transmission rate for a user $u \in \mathcal{U}_n$, transmitting input data for LLM agents to the LLM service provider, can be expressed by the following equation:
% \begin{equation}
% r_{u,n} = B_n \log_2 \left(1 + \frac{g_{u, n}p_u}{\sum_{j \in \mathcal{U}_n \backslash \{u\}} g_{j,n} p_j + \sigma^2}\right),
% \end{equation}
% where $p_u$ represents the transmit power of user $u$, and $\sigma^2$ is the noise power of the additive white Gaussian noise (AWGN).  Each ES is connected to cloud data center through the terrestrial core network, which supports the transmission of data at a fixed rate $r^C_n$ for ES $n$. This configuration ensures robust and efficient data flows between mobile users and cloud services, facilitating high-performance access to LLM agents without intervention.

To optimize the serving of LLMs in mobile edge networks, we propose a joint model caching and inference framework that enables ESs to deploy LLMs for local execution and offload inference requests to cloud data centers for remote execution. In the framework, ESs can determine the cached models and offloaded requests based on their local caching and offloading strategies. Here, the binary variable $a_{n, i, m}^t \in \{0,1\}$ specifies if LLM $m$ for agent $i$ is cached at ES $n$ during time slot $t$, and the continuous variable $b_{n, i, m}^t \in [0,1]$ indicates the fraction of inference requests $i$ for LLM $m$ being executed at ES $n$ at that time slot. The set $\mathbf{a}^t_n = \{a^t_{n,1,1}, \ldots, a^t_{n,I,M}\}$ encapsulates the model caching decisions at ES $n$. Similarly, the inference offloading decision for ES $n$ is captured by $\mathbf{b}^t_n = \{b^t_{n,1,1}, \ldots, b^t_{n, I, M}\}$. 

When ESs are constrained by local resources, user-requested LLM agents can be executed locally if the necessary models are already cached in the GPUs. Let $G_n$ represent the GPU memory capacity at ES $n$. The decisions regarding which models to cache, denoted as $\mathbf{a}^t_n$, must adhere to the GPU memory limits for each time $t$ for every ES, described by:
\begin{equation}
\sum_{i \in \mathcal{I}} \sum_{m \in \mathcal{M}} a_{n, i, m}^t s_{m} \leq G_n.
\end{equation}
This constraint imposes that ESs cannot accommodate all LLMs on their GPUs simultaneously. When the LLMs are cached in GPUs, the LLM agents can be served by the ESs. The local inference execution of LLMs at ES $n$ is further limited by
\begin{equation}
b_{n, i, m}^t \mathbb{I}(R_{n, i, m}^t > 0) \leq a_{n, i, m}^t, \forall i \in \mathcal{I}, m \in \mathcal{M},
\end{equation}
where $\mathbb{I}(\cdot)$ is the identity function, and $\mathbb{I}(R_{n, i, m}^t > 0)$ indicates that there are active requests for agent $i$ for LLM $m$ at ES $n$.

Furthermore, the computing resources for executing LLMs are limited by the energy capacity $E_n$ at ES $n$, which can be represented as
\begin{equation}
\sum_{i \in \mathcal{I}} \sum_{m \in \mathcal{M}} e_m a_{n, i, m}^t (1 - b_{n, i, m}^t) R_{n, i, m}^t \leq E_n.
\end{equation}
In contrast, cloud data centers are generally assumed to have no significant constraints on GPU memory or computing capacity for executing LLMs~\cite{zhao2022edgeadaptor}.

\subsection{CoT Inference Model}

The LLMs utilize vast datasets and complex neural network architectures (like transformers) to model the joint probability distribution of sequences of text. However, rather than modeling the complete joint distribution directly (which is computationally infeasible due to its complexity and dimensionality), they approximate the marginal distribution. This is the distribution of individual tokens or phrases given their preceding context in a sequence.

\begin{definition}[LLMs as marginal approximations \cite{Yun2020Are}]
LLMs can be considered as a universal density approximator specifically for the domain of natural language. Let $d_{i,0}$ be the task description prompt and $ D_i = \{d_{i,0}, d_{i,1}, \ldots, d_{i,k}\} $ be the sequence of messages in LLM agent $i$. The learned distribution of LLM $m$ can be represented as $ p_m $ and the true distribution of the sequences of messages is represented as $\hat{q}$.
For all sequences of messages $ D_i = \{d_{i,0}, d_{i,1}, \ldots, d_{i,k}\} $ of at most $ w_m $ tokens, which is limited by the context window, the learned distribution $ p_m(D_i) $ follows the general product rule of probability
        $p_m(D_i) = p_m(d_{i,0}, d_{i,1}, \ldots, d_{i,k}) = p_m(d_{i,0}) p_m(d_{i,1} | d_{i,0}) \cdots p_m(d_{i,k} | d_{i,0}, d_{i,1}, \ldots, d_{i,k-1})$,
which is a good approximation of the true distribution $ \hat{q}(D_i) $, i.e., $p_m(D_i) \approx \hat{q}(D_i)$.
% The model is "universal" in that for any given true probability distribution $ q(x) $ of language data (the marginal distribution of natural language), and for any non-zero level of precision $ \epsilon $, there exists a configuration of parameters $ \Lambda $ such that the absolute difference between the true distribution $ q(x) $ and the model's approximation $ p_{\Lambda}(x) $ is less than $ \epsilon $ for all $ x $, which can be expressed as:
% \begin{equation}
% \forall q(x), \forall \epsilon > 0, \exists \Lambda : |q(x) - p_{\Lambda}(x)| \leq \epsilon \ (\forall x),
% \end{equation}
% where $ \Lambda $ represents the set of all model parameters.
\end{definition}

Finally, LLM $m$ generates an answer by recursively predicting the sequence of the next tokens from the learned distribution $ p_m $, conditioned on the concatenation of the task description prompt $d$ and the tokens sampled so far. The approximation leverages the product rule of probability to decompose the joint probability of the sequence $ D_i $ into the product of conditional probabilities of each example $ d_{i,j} $ given the previous examples $ \{d_{i,0}, d_{i,1}, \ldots, d_{i,j-1}\} $.

In CoT reasoning, the CoT example $E_i$ generated from an LLM is generated based on a true context $c^*$ and a true intention $\theta^*$~\cite{tutunov2023can}. These elements represent the underlying conditions and objectives that guide the generation of the text or response by LLMs. The ambiguity of a chain of thought example is quantified by how likely it is that the true context and intention can be accurately inferred from the example. This likelihood is termed as the complement of the ambiguity.

\begin{definition}[$\epsilon$-Ambiguity~\cite{jiang2023latent}]
The ambiguity of the chain, given the context and intention, is defined as 
$\hat{q}(c^*, \theta^* | E_i) = 1 - \epsilon(E_i).$
This equation states that the likelihood $ \hat{q} $ of correctly identifying the context $ c^* $ and intention $ \theta^* $ from the example $ E_i $ is $ 1 - \epsilon(E_i) $, where $ \epsilon(E_i) $ represents the measure of ambiguity. The smaller the value of $ \epsilon(E_i) $, the less ambiguous the example is, meaning the context and intention are more clearly communicated or inferred from $ E_i $.
\end{definition}
Specifically, a lower $\epsilon(E_i)$ value indicates higher clarity and lower ambiguity, making it easier to understand the intended message or instruction embedded in $ E_i $. Conversely, a higher $\epsilon(E_i)$ indicates that the example is more ambiguous, and the true context and intention are harder to discern.
To enhance the relevance and coherence of LLM agents through SC-CoT prompting, LLMs utilize a multi-path reasoning approach, where the final response is determined by reaching a consensus among various reasoning paths. This SC-CoT method samples multiple reasoning trajectories and chooses the most consistent outcome, thereby improving the robustness of the inference process. When LLMs perform inference, provided with a task description prompt $ d_0 $, LLM $m$ can generate $ J_{i,m} $ reasoning paths for agent $i$ by recursively predicting subsequent tokens from the learned distribution $ p_m $, based on the concatenation of $ d $ and the previously sampled tokens.
    
\begin{definition}[SC-CoT Reasoning~\cite{wang2022self}]
Each reasoning path $ \rho_j $ results in a candidate output $ o_j $, representing the final answer derived from the reasoning path. Thus, we obtain $ J_{i,m} $ candidate outputs $ \{A_1, \ldots, A_J\}$ required by LLM $m$ of agent $i$. For every reasoning path $ j $, LLM $ m $ receives an SC-CoT example $ E_i $ of varying lengths $ c_{j,i} $. This example $ E_i = \{e_{i,0}, \ldots, e_{i,k}\} $ consists of multiple thoughts, where each thought $ e_{i,k} $ is made up of $ k_i $ tokens, representing an individual step in the reasoning process. The marginal probability of any output $o$ is computed by integrating all the different reasoning paths as $P(o|d) = \sum_{j=1}^{J_{i,m}} P(\rho_j, o|d)$.

By considering all the reasoning paths, the most consistent answer $ o^* $ is the one that maximizes the marginal probability $o^*= \arg\max_{o} \sum_{j=1}^{J_{i,m}} P(\rho_j, o | d) = \arg\max_{o} \sum_{j=1}^{J_{i,m}} \mathbb{I}(o_j = o)$, where $\mathbb{I}(\cdot)$ is the identity function.
When each reasoning path is equally likely, this can be simplified to finding the most frequent answer among the candidate outputs.
\end{definition}

To derive the expected difference bound of language ambiguity, we first make the following assumptions~\cite{Yun2020Are, tutunov2023can, jiang2023latent}:
\begin{enumerate}
    \item It is a uniform distribution for the true contexts $c^*$.
    \item The example $E_i$ generated from $ \theta^* $ with  the true context $ c^*$ is bounded by the ambiguity measure, i.e., $ \epsilon(E_i) = \hat{q}(c^*, \theta^*|E_i) \leq \sigma $, where $ \sigma \in [0, \frac{1}{2}] $.
    \item For examples $ E_i$ derived from true intentions $\theta^*$ with the true context $ c^*$, the related ambiguity measure $ \epsilon(E_i)$ diminishes as the length of inference sequences increases, i.e., $ \lim_{l \to \infty} \epsilon(E_i) = 0 $.
\end{enumerate}

% \begin{assumption}
%     The prior probability distribution of true contexts $c^*$ is assumed to be uniform across all possible contexts.
% \end{assumption}

These assumptions posit that each possible context $ c^* $ from which language examples are generated is equally likely to occur before considering any specific examples or data. In probabilistic terms, a uniform prior distribution means that no contextual scenario is inherently more probable than any other before any additional information is taken into account. Therefore, they can make the skewness parameter $ \gamma_n(c^*) = \sup_{c\in C}\frac{\hat{q}(c^\star)}{\hat{q}(c_n)}$ is equal to 1 for context $c_n$ belonged to ES $n$, which helps in estimating the difference in ambiguity measurements between learned distribution $p_m$ and true distribution $\hat{q}$.

\subsection{Performance Convergence Analysis for CoT Prompting}

SC-CoT prompting can improve the capability of LLMs to handle complex reasoning, by merging multiple reasoning paths with a consensus mechanism, which ensures more reliable responses. In mobile edge networks, where latency and resources are restricted, SC-CoT is crucial to let LLMs explore various thought paths, enhancing their ability to refine responses based on broader context, thus boosting accuracy and decision-making. This approach mimics human-like problem-solving processes by breaking down complex tasks into simpler, sequential steps. CoT reasoning is particularly effective in enhancing the model's ability to tackle complex reasoning, arithmetic, commonsense, and symbolic reasoning tasks by providing a transparent and step-wise articulation of thought processes.

\begin{lemma}
%     Given a set of CoT examples $ E_i $ of varying lengths $ c_i $, generated from the intention $ \theta^* $ with the optimal context $ c^* $ derived from the distribution $ q_m(c) $ under Assumption 1, and considering $ d_{i,0} $ as the initial input message or task sampled from $ q(\theta^*) $, then for any sequence of messages $ D_i $, the likelihood ratio of the model’s output distribution $ p_m(D_i|d_{i,0}, E_i) $ to the   truth distribution $ q(D_i|d_{i,0}, c^*) $ is bounded by
% \begin{equation}
%     \left|\frac{p_m(D_i|d_{i,0}, E_i) - q(D_i|d_{i,0}, c^*)}{q(D_i|d_{i,0}, c^*)}\right| \leq \eta \prod_{y=1}^{c_i} \frac{\epsilon(E_{i,y})}{1 - \epsilon(E_{i,y})},
% \end{equation}
% where $ \eta = \frac{2\epsilon(d_{i,0})}{1-\epsilon(d_{i,0})} $ is a factor that depends on the ambiguity of the initial input $ d_{i,0} $.
% Considering a collection of $ c_i $ varying-length CoT examples, which are generated from the intention $ \theta^* $ with the optimal context $ c^* $ sampled from $ q_m(c) $ that satisfies Assumption 1. Furthermore, let $ d_{i,0} $ be the input message or task sampled from $ q(\cdot|\theta^*_0) $, which is generated from $ \theta^*_0 $ sampled from $ q_m(\cdot|c^*) $. Then, for any sequence of messages $ D_i $, we have:
% \begin{equation}
% |p_m(D_i|d_{i,0}, E_i) - \hat{q}(D_i|d_{i,0}, c^*)| \leq \eta \prod_{y=1}^{c_i} \frac{\epsilon(E_i,y)}{1 - \epsilon(E_i,y)},
% \end{equation}
% where $ \eta = 2 \frac{\epsilon(d_{i,0})}{1 - \epsilon(d_{i,0})} $ depends on the ambiguity of the input.
Given a set of examples $ c_i $ with different lengths, created from the true intention $ \theta^* $ and the true context $ c^* $ drawn from $q_m(c)$, let $ d_{i,0} $ be the initial message or task, which is derived from $ q(\cdot|\theta^*_0) $ and generated from $ \theta^0 $ sampled from $ q_m(\cdot|c^*) $. For any sequence of messages $D_i$, the following holds:
\begin{equation}
|p_m(D_i|d_{i,0}, E_i) - \hat{q}(D_i|d_{i,0}, c^*)| \leq \eta \prod_{y=1}^{c_i} \frac{\epsilon(E_i,y)}{1 - \epsilon(E_i,y)},
\end{equation}
where $ \eta = 2 \frac{\epsilon(d_{i,0})}{1 - \epsilon(d_{i,0})} $ is related to the ambiguity of the input.
\end{lemma}

The proof of Lemma 1 can be found in~\cite{xu2024cached} and this lemma essentially supports the concept of CoT by establishing a mathematical relationship that quantifies how the ambiguity in input influences the accuracy of the CoT reasoning process. The product term $\prod_{y=1}^{c_i} \frac{\epsilon(E_{i,y})}{1 - \epsilon(E_{i,y})}$ reflects the cumulative effect of errors in each step of reasoning on the overall reasoning process, emphasizing the importance of minimizing error at each reasoning step to ensure the overall reliability of the CoT reasoning process. This supports the practical implementation of CoT reasoning by highlighting the need for precise and less ambiguous reasoning steps in complex problem-solving tasks.

Furthermore, SC-CoT~\cite{wang2022self} reasoning enhances the typical chain-of-thought reasoning process by generating multiple reasoning paths and selecting the most consistent answer from those paths. This approach is based on the idea that when multiple reasoning paths converge on the same answer, it increases confidence in the correctness of that answer. This method is particularly useful for complex reasoning tasks, where a single reasoning path might be prone to errors or bias. In SC-CoT, each reasoning path $ j $ within a CoT example $ i $ contributes independently to the formulation of the output. This formulation allows for the incorporation of diversity in reasoning strategies and the aggregation of these paths to determine the most consistent and probable outcomes.

To effectively represent the features of multiple reasoning paths in SC-CoT reasoning within a theoretical framework, we extend lemma 1 to explicitly incorporate these features and their impact on the overall reasoning process as follows.

\begin{theorem}
Given a set of $J_{i,m}$ SC-CoT reasoning paths sampled as $ \{R_1, \ldots, R_j, \ldots, R_{J_{i,m}}\}$ with varying lengths, the final output is sampled from making consensus to obtain the most frequent answer under the true intention $ \theta^* $ with the true context $c^*$ drawn from $ q_m(c)$, let $ d_{i,0} $ be the initial input message or task, which is sampled from $ q(\cdot|\theta^*_0) $ and generated from $ \theta^*_0 $ drawn from $ q_m(\cdot|c^*) $. For any sequence of messages $ D_j $, and with $J_{i,m}$ reasoning paths sampled as $ \{R_1, \ldots, R_j, \ldots, R_{J_{i,m}}\} $, the following holds:
\begin{equation}
|p_m(D_{i,j}|d_{i,0}, E_{i,j}) - \hat{q}(D_{i,j}|d_{i,0}, c^*)| \leq \eta \prod_{j=1}^{J_{i,m}} \frac{\epsilon(E_{i,j})}{1 - \epsilon(E_{i,j})},
\end{equation}
where $ \eta = 2 \zeta_i \frac{\epsilon(d_{i,0})}{1 - \epsilon(d_{i,0})} $, which is dependent on the ambiguity of the input.
\end{theorem}
\begin{proof}
    Starting from $ p_m(D_{i,j}|d_{i,0}, E_{i,j}) $, we sample $ j $ reasoning paths and marginalize over these paths
\begin{equation}
p_m(D_{i,j}|d_{i,0},E_{i,j}) = \frac{\sum_{j=1}^{J_{i,m}}  \hat{q}(D_{i,j}, E_{i,j}, R_j)}{\sum_{j=1}^{J_{i,m}}  \hat{q}(d_{i,0}, E_{i,j}, R_j)},
\end{equation}
\begin{equation}
= \frac{ \sum_{j=1}^{J_{i,m}}  (\hat{q}(D_{i,j} \setminus \{d_{i,0}\}, E_{i,j}, R_j, c^*) + \Lambda_j)}{1 + \sum_{j=1}^{J_{i,m}}  \Upsilon_j},
\end{equation}
where $ \Lambda_j $ and $ \Upsilon_j $ are given by:
\begin{equation}
\Lambda_j = \sum_{c \neq c^*} \frac{\hat{q}(D_{i,j}, E_{i,j}, R_j, c)}{\hat{q}(d_{i,0}, E_{i,j}, R_j, c^*)}
\end{equation}
and
\begin{equation}
    \Upsilon_j = \sum_{c \neq c^*} \frac{\hat{q}(d_{i,0}, E_{i,j}, R_j, c)}{\hat{q}(d_{i,0}, E_{i,j}, R_j, c^*)}.
\end{equation}
Using the ambiguity measure for CoT examples, we establish the bounds on $ \Lambda_j $ and $ \Upsilon_j $ as
\begin{equation}
\Lambda_j, \Upsilon_j \leq \gamma^{c_{j,i}}_n (c^*) \frac{\epsilon(d_{i,0})}{1 - \epsilon(d_{i,0})} \prod_{y=1}^{c_{j,i}} \frac{\epsilon(E_{i,j})}{1 - \epsilon(E_{i,j})}
\end{equation}

Finally, by combining Eqs. (14)-(18), we can get
\begin{equation}
\begin{aligned}
    |p_m(D_{i,j}|&d_{i,0}, E_{i,j}) - \hat{q}(D_{i,j}|d_{i,0}, c^*)| \\&= |\frac{\sum_{j=1}^{J_{i,m}}  (\hat{q}(D_{i,j} \setminus \{d_{i,0}\}, E_{i,j}, R_j, c^*) + \Lambda_j)}{1 + \sum_{j=1}^{J_{i,m}}  \Upsilon_j} \\ & - \hat{q}(D_{i,j}|d_{i,0}, c^*)| \leq \eta \prod_{y=1}^{c_{j,i}} \frac{\epsilon(E_{i,j})}{1 - \epsilon(E_{i,j})}
\end{aligned}
\end{equation}
where $ \eta = 2 \zeta_i \frac{\gamma^{J_{i,m}}_n (c^*) \epsilon(d_{i,0})}{1 - \epsilon(d_{i,0})} $, following the first assumption, indicating $ \gamma_n(c^*) = 1 $ and thus $ \eta = 2 \zeta_i \frac{\epsilon(d_{i,0})}{1 - \epsilon(d_{i,0})} $.
\end{proof}

The theorem implies that the accuracy and robustness of the model's output are closely tied to the clarity and specificity of the context and intentions it is trained with. If the examples and initial inputs are unambiguous, the model's predictions will be closer to the true context, resulting in more accurate and reliable outputs. The outer product overall reasoning paths represent the compounded benefit of exploring multiple reasoning strategies. By aggregating across diverse paths, the method leverages redundancy and diversity to enhance the likelihood of converging on the correct answer, especially in complex reasoning scenarios where single-path strategies might fail. This theorem shows how the divergence from the true distribution is not only a function of the ambiguities of individual reasoning steps but also the initial conditions set by $ d_{i,0} $. Effective reduction in initial ambiguity and careful management of step-wise ambiguities are crucial for enhancing the accuracy and reliability of SC-CoT systems.

\begin{corollary}
When examining SC-CoT examples $ E_{i,j} $, for any chosen $ \sigma $ within the interval $ \left[0, \frac{1}{2}\right] $, a length threshold $ k^*_{i,\sigma} \in \mathbb{N} $ exists. For any length $ k_i $ that is greater than or equal to $ k^*_{i,\sigma} $, the following condition holds true: $\epsilon(E_{i,j}) \leq \sigma$.

\end{corollary}
\begin{proof}
    By selecting $ \sigma \in \left[0, \frac{1}{2}\right] $, SC-CoT examples $ E_{i,j} $ following Assumption 3 have the approximation that $\lim_{l \to \infty} \epsilon(E_{i,j}) = 0$.
Then, there exists $ k^*_{i,\sigma} \in \mathbb{N} $ such that for any $ k_i \geq k^*_{i,\sigma} $, the inequality $ \epsilon(E_{i,j}) \leq \sigma $ holds.
\end{proof}
This proposition supports the SC-CoT approach by providing a theoretical foundation for the reduction of ambiguity with sufficiently long CoT examples, thereby enhancing the reliability and accuracy of the model's outputs.

In SC-CoT, the number of thoughts $ K^t_{n,i,m} $ is accumulated during the inference process for cached LLM $m$ to serve agent $i$ at ES $n$.  Depending on the caching decision $a_{n, i, m}^t$ and offloading decision $b_{n, i, m}$, the batch of inference requests executed as ES $n$ can be calculated as $\delta_{n, i, m, j}^t = a_{n, i, m}^t (1-b_{n, i, m}^t) R_{n, i,m}^t J_{i,m} k_i$ for agent $i$ and LLM $m$, where $k_i$ is the size of each thought for agent $i$ estimated via Corollary 1. Therefore, the equation to reflect the aggregation of multiple sampled paths is given as
\begin{equation}
K^t_{n,i,m} = \begin{cases} 
0, & t = 0, \\
a^t_{n,i,m} \left( K^{t-1}_{n,i,m} + \delta^t_{n,i,m} \right), & \text{otherwise},
\end{cases}
\end{equation}
where $ \delta^t_{n,i,m} = \sum_{j=1}^{J_{i,m}} \delta^t_{n,i,m,j}$ is the total calculated tokens over $J$ reasoning paths.

Similarly, age of thoughts (AoT) $\kappa^t_{n,i,m}$ reflect the multiple sampled paths~\cite{xu2024cached} and their aggregated contributions as
\begin{equation}
\kappa^t_{n,i,m} = \begin{cases} 
0, & t = 0, \\
a^t_{n,i,m} \left\{\kappa^{t-1}_{n,i,m} + \zeta_i \delta^t_{n,i,m} - \Delta^t_{i,m} \right\}^+, & \text{otherwise,}
\end{cases}
\end{equation}
where $\zeta_i$ is the consensus factor of SC-CoT reasoning for LLM agent $i$ and $ \Delta^t_{i,m} $ is the vanishing factor of thoughts.

To reflect the SC-CoT process, we modify the performance Eq. (5) to include the most consistent answer selection mechanism as
\begin{equation}
A^t_{n,i,m} = \alpha_{i,m} \log \left( \frac{1}{\beta^{\kappa^t_{n,i,m}}} \right),
\end{equation}
where $ \alpha_{i,m} $ is the zero-shot accuracy of LLM $ m $ for LLM agent $ i $,  $ \beta $ is the reasoning gain factor reflecting the performance improvement due to self-consistency, and $ \kappa^t_{n,i,m} $ represents the adjusted AoT considering the most consistent answers selected from multiple reasoning paths.

\subsection{Cost Structure and Problem Formulation}

\subsubsection{Cost Structure}

Edge inference costs are incurred when LLM agents are processed at ESs. These costs include switching costs from loading and evicting models in GPU memory, transmission costs for sending input prompts and results between users and ESs, computation costs for processing data at ESs, and accuracy costs due to performance gaps from using alternative LLMs. First, the switching cost arises when switching models in the GPU memory at the ESs~\cite{zhao2022edgeadaptor}. It includes the latency and hardware wear associated with loading and evicting operations, i.e.,
\begin{equation}
l_{swi}^n(\textbf{a}^t_n) = \sum_{i \in \mathcal{I}} \sum_{m \in \mathcal{M}} \lambda \cdot  \mathbb{I}(a_{n,i,m}^t > a_{n,i,m}^{t-1}),
\end{equation}
where $ \lambda $ is the unit cost coefficient for loading model parameters, and $ \mathbb{I}(\cdot) $ is an indicator function.

The transmission cost involves transmitting input prompts and inference results between users and ESs, accounting for the increased data transmission due to multiple reasoning paths, which can be calculated as
\begin{equation}
l_{com}^n(\textbf{a}^t_n, \textbf{b}^t_n) = \sum_{i \in \mathcal{I}} \sum_{m \in \mathcal{M}} R_{n,i,m}^t  l_{n,i} d_i  (1-b_{n,i,m}^t),
\end{equation}
where $ l_{edge}^{n} $ is the unit transmission cost for edge network access, and $ d_i $ is the input data size.

The computation cost of executing inference requests at ES $n$, considering the multiple reasoning paths generated and evaluated by using SC-CoT prompting, i.e.,
\begin{equation}
l_{cmp}^n(\textbf{a}^t_n, \textbf{b}^t_n) = \sum_{i \in \mathcal{I}} \sum_{m \in \mathcal{M}} \sum_{j=1}^{J_{i,m}} \frac{\delta_{n,i,m, j}^t e_m}{f_n},
\end{equation}
where $ \delta_{n,i,m,j}^t $ is the total computation token for the $ j $-th reasoning path, $ e_m $ is the energy consumption per token, and $ f_n $ is the computing capacity of ES $n$.

The accuracy cost is incurred when there is a performance gap due to the use of alternative LLMs, considering the consensus mechanism from multiple reasoning paths as
\begin{equation}
l_{acc}^n(\textbf{a}^t_n, \textbf{b}^t_n) = \sum_{i \in \mathcal{I}} \sum_{m \in \mathcal{M}} \bar{A}_{n,i,m}^t R_{n,i,m}^t a_{n,i,m}^t (1 - b_{n,i,m}^t),
\end{equation}
where $\bar{A}_{n, i,m}^t = \frac{1-\alpha_{i,m}}{\kappa^t_{n,i,m} \log(1/\beta_{i, m})}$ is the accuracy cost of executing LLM $m$ at ES $n$ for agent $i$.

Overall, the total edge inference cost is the sum of these components
\begin{equation}
\begin{aligned}
    L_n^t(\textbf{a}^t_n, \textbf{b}^t_n) = l_{swi}^n(\textbf{a}^t_n) &+ l_{com}^n(\textbf{a}^t_n, \textbf{b}^t_n) \\&+ l_{cmp}^n(\textbf{a}^t_n, \textbf{b}^t_n) + l_{acc}^n(\textbf{a}^t_n, \textbf{b}^t_n).
\end{aligned}
% L_n^t(\textbf{a}^t_n, \textbf{b}^t_n) = l_{swi}^n(\textbf{a}^t_n) + l_{com}^n(\textbf{a}^t_n, \textbf{b}^t_n) + l_{cmp}^n(\textbf{a}^t_n, \textbf{b}^t_n) + l_{acc}^n(\textbf{a}^t_n, \textbf{b}^t_n).
\end{equation}

When ES $n$ cannot handle all requests, its uncompleted inference requests are offloaded to cloud data centers, incurring cloud inference costs. This includes the costs associated with offloading requests to the cloud and processing them remotely. The cloud inference cost is calculated based on the unit processing cost at the cloud data center and the size of the data processed. The cloud inference cost includes the cost of offloading and processing these requests, accounting for the additional computation and data transmission due to multiple reasoning paths, which can be calculated as
\begin{equation}
L_C^t(\textbf{a}^t_n, \textbf{b}^t_n) = \sum_{n \in N} \sum_{i \in \mathcal{I}} \sum_{m \in \mathcal{M}} \sum_{j=1}^{J_{i,m}} l_0^{m,j} b_{n,i,m}^t R_{n,i,m}^t,
\end{equation}
where $ l_0^{m,j} $ is the unit processing cost at the cloud data center for executing reasoning path $j$ of LLM $m$.

The total cost for provisioning LLM agents for an ES combines both edge and cloud inference costs over a specific period. By summing these costs over time, the total cost provides a comprehensive measure of the resources and quality of LLM agents involved in operating LLM agents, helping to ensure efficient resource utilization and cost minimization while maintaining high perception and reasoning performance. Therefore, the total cost for provisioning LLM agents at ES $ n $ is the sum of the edge inference cost and the cloud inference cost over a period $ T $ as

\begin{equation}
L_{total}^n = \frac{1}{T} \sum_{t = 1}^{T} \left(L_n^t(\textbf{a}^t_n, \textbf{b}^t_n) + L_C^t(\textbf{a}^t_n, \textbf{b}^t_n)\right).
\end{equation}
However, minimizing the total system cost under hardware resource constraints and context window constraints is NP-hard as the decision variables include discrete and continuous variables while the constraints are not linear.

\section{Test-time Reinforcement Learning-based Model Caching and Inference Offloading Algorithm}

To optimize long-context LLM serving in mobile edge networks proactively, we propose the T2DRL algorithm by leveraging TTT models in determining model caching and inference offloading decisions. Before deployment of the algorithm, the environment needs to be formulated as an MDP consisting of state, observation, action, and reward.
\subsection{Markov Decision Process}

\paragraph{Observation Space}
During the serving of LLM agents in the SC-CoT reasoning model, the state space captures the comprehensive status of the system, while the observation space is a subset of state space. Therefore, the observation $ O_n^t \in\mathcal{O}$ of ES $n$ at time $t$ is a combination of the user request information, cached model status, and model context utilization information, which can be formally defined as
$O_n^t := \left( \textbf{a}^t_n, R_n^t, \kappa_{n}^t \right)$,
where $ R_n^t$ represents the set of all user request information; $\textbf{a}_{n}^t$ represents the set of all cached model statuses, indicating whether a model is cached at a specific ES and the number of tokens utilized; and $\kappa_{n}^t = \{\kappa_{n,1,1}^t,\dots,\kappa_{n, I, M}^t\}$ represents the set of all model context utilization information, describing the current AoT for each cached LLM at ESs.

\paragraph{Action Space}
The action $X^t_n \in \mathcal{X}$ in action space for the SC-CoT framework encompasses both model caching decisions $ \textbf{a}^t_n $ and inference offloading decisions $\textbf{b}^t_n$, which can be defined as $X^t_n := \left( \textbf{a}^t_n, \textbf{b}^t_n\right)$.

% In the action space, model caching decisions $ \textbf{a}^t_n $ are binary variables indicating whether a specific model is cached at a particular ES, directly impacting the response time and efficiency of handling user requests. Meanwhile, inference offloading decisions $\textbf{b}^t_n$ are continuous variables that determine the proportion of user requests offloaded from an ES to a cloud data center, helping to balance the load and optimize resource utilization across the network.

\paragraph{Reward Function}
The reward function is designed to minimize the total system cost, which includes edge inference costs and cloud inference costs. The goal is to formulate a reward function that encourages actions leading to lower overall costs while maintaining high-quality LLM agents. To formulate the reward function, we aim to minimize the total system cost. The reward at time $ t $, denoted as $ r^t_n $, is defined as the negative of the total system cost, ensuring that lower costs yield higher rewards as $r^t_n(O_n^t, X_n^t) = -L_{total}^n(\textbf{a}^t_n, \textbf{b}^t_n)$, which can effectively guide the algorithm towards minimizing the total system costs, ensuring efficient and cost-effective serving of LLM agents in the joint model caching and inference offloading framework.
% \paragraph{Probability Transition Function}
% \paragraph{Value Function}
\subsection{Test-time Deep Reinforcement Learning}

\begin{figure}[t]
    \centering
    \includegraphics[width=1\linewidth]{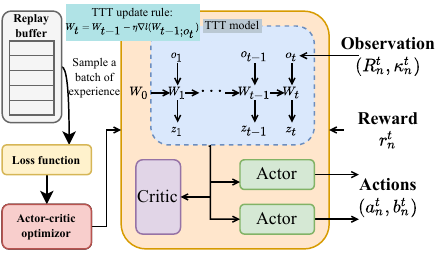}
    \caption{The T2DRL algorithm utilizes the test-time training (TTT) model in the actor-critic network.}
    \label{fig:algorithm}
\end{figure}

As shown in Fig.~\ref{fig:algorithm}, the T2DRL algorithm switches between exploration and exploitation by engaging with its surroundings and refining an alternative target function through random gradient ascent. The process initiates with the data collection phase, where the policy $\pi_n$ is executed in the environment to gather a set of trajectories  $(O_n^t, X_n^t, r^t_n, O_n^{t+1})$. Each step involves executing the current policy $\pi$ to determine the action and observing the next observations and rewards. The DRL algorithm leverages a batch of trajectories comprised of sequences of states, actions, and rewards in the replay buffer, which is used to calculate advantages for policy optimization.

At the beginning of each episode, the T2DRL policy $\pi_n$ consists of initialized weights $W_0$ for TTT models.  For each input observation  $O^t_{n}$ in the sequence, the policy $\pi$ predicts an intermediate states $ z^t_n  $ as  $ z^t_n = \pi_n(O^t_{n}; W_t) $. The weights  $W_t$ are then updated using a gradient descent step on a self-supervised loss  $\ell$, following  $ W_t = W_{t-1} - \eta \nabla \ell(W_{t-1}; O^t_{n})$, where the self-supervised loss can be calculated as $\ell(W; O^t_{n}) = \left \|\pi_n(\tilde{O}^t_{n}; W) - O^t_{n}\right \|^2$ and $\tilde{O}^t_{n}$ is the approximation of true state of environment. This update mechanism enables continuous learning during testing phases, allowing the model to adapt its parameters based on the immediate feedback.  This process is repeated for each token in the observation sequence.

Next, advantage estimates  $ \tilde{A}_t  $ are computed using Generalized Advantage Estimation, which effectively reduces the variance of the policy gradient estimates by leveraging the temporal difference residuals.  The advantage estimates $ \tilde{A}_t  $ is calculated as the sum of discounted temporal differences  $\delta_t$. 
i.e., $\tilde{A}_t = \delta_t + (\gamma \lambda) \delta_{t+1} + \cdots + (\gamma \lambda)^{T-t+1} \delta_{T-1}$, where  $ \delta_t = r_t + \gamma V(s_{t+1}) - V(s_t)  $ and the state value function  $ V_\theta(s_t) = \mathbb{E}_\pi \left[ \sum{t=0}^{\infty} \gamma^t r^t_n \mid O^0 = O \right]  $, which estimates the expected return from state  $ O^t  $ under the current policy  $\pi(\theta)$. 

Following self-supervised loss and advantage estimation, the TTT model-based policy can output the action $X_n^t = \pi(O_n^t; \theta; W_t)$ for the observed. The model updates its hidden state using the trajectory data and performs test-time training by updating its weights with a self-supervised loss function $\ell$. This adaptation process enables the model to tune its parameters specifically to the nuances of the encountered test data, enhancing its predictive accuracy and responsiveness. Furthermore, the policy of T2DRL is then updated by defining the PPO surrogate loss with the clipped objective.
% \begin{equation}
% \mathcal{L}_{\text{CLIP}}(\theta) = \hat{E}_t \left[ \min \left( r_t(\theta) \hat{A}_t, \text{clip}(r_t(\theta), 1 - \epsilon, 1 + \epsilon) \hat{A}_t \right) \right],
% \end{equation}
% where  $ r_t(\theta) = \frac{\pi_\theta(a_t | s_t)}{\pi_{\theta_{\text{old}}}(a_t | s_t)}  $ represents the probability ratio of the new policy to the old policy, ensuring that updates do not deviate excessively from previous policies, thereby maintaining training stability. Then, the combined objective function can be represented as
% \begin{equation}
% \mathcal{L}_{\text{CLIP+VF+S}}(\theta) = \hat{E}_t \left[ \mathcal{L}_{\text{CLIP}}(\theta) - c_1 \mathcal{L}_{\text{VF}}(\theta) + c_2 S_{\pi_\theta} \right],
% \end{equation}
% where  $ \mathcal{L}_{\text{VF}}(\theta) = (V_\theta(s_t) - V_t^{\text{target}})^2  $ and  $ S_{\pi_\theta}  $ is the entropy bonus, which encourages policy exploration by maximizing policy randomness. This comprehensive function is optimized through several epochs of minibatch optimization using stochastic gradient ascent to effectively balance policy performance, value accuracy, and exploration diversity.

Finally, the policy parameters  $ \theta$ and TTT model parameters $W$ are updated with the optimized values. This iterative process continues until convergence or for a fixed number of iterations. By utilizing TTT models in the actor-critic network, the T2DRL algorithm can adapt its learned policy and model parameters during test phases dynamically, enhancing the system's ability to efficiently manage model caching and inference offloading. This optimization leads to improved resource utilization and minimized operational costs while ensuring high-quality LLM agents, making the T2DRL framework a robust and adaptive solution for long-context LLM serving in mobile edge networks. % The detailed training procedures are outlined in Algorithm 1.

\section{Market Formulation and Double Dutch Auction Mechanism Design}

For serving long-context LLMs at mobile edge networks, the auction market is modeled as a dynamic interaction between resource-limited edge servers (sellers) and user devices (buyers). Formulating a long-context LLM serving market at mobile edge networks is essential to efficiently allocate limited edge resources, optimize costs, and handle dynamic user demands. It provides a structured framework for balancing trade-offs between computational efficiency, accuracy, and latency while ensuring fairness and transparency in resource sharing. By enabling real-time adaptation and incentivizing truthful participation, the market framework maximizes social welfare and scalability, supporting diverse applications in resource-constrained edge environments. In the long-context LLM serving market, there is an auctioneer who coordinates the market to optimize the allocation of resources and inference tasks.

\subsection{Long-context LLM Serving Market Formulation}

The long-context LLM serving market consists of three primary entities:
\begin{itemize}
    \item \textbf{Buyers (User Devices):} Each buyer $i$ submits a buy-bid, $k^t_i$, representing the maximum price they are willing to pay for LLM inference services. The valuation $v^b_i$ depends on the required quality $Q_i$ and resource demand $R_i$, given as:
    \begin{equation}
    v^b_i = Q_i l_{acc}^i,
    \end{equation}
    where $Q_i$ reflects the inference quality coefficient of buyer $i$.

    \item \textbf{Sellers (Edge Servers):} Each seller $j$ submits a sell-bid, $o^t_j$, indicating the minimum price they are willing to accept. This valuation $v^s_j$ is based on resource costs, including memory and computation, defined as:
    \begin{equation}
    v^s_j = {G_j} (l_{swi}^j + l_{com}^j + l_{cmp}^j),
    \end{equation}
    where $G_j$ is server $j$ resource utilization coefficient.

    \item \textbf{Auctioneer:} Operates the DDA mechanism, coordinating buyer and seller clocks to identify the market-clearing price, $p^*$.
\end{itemize}

\subsection{Auction Mechanism Design}

Double Dutch auction excels in addressing the computational demands of multi-path reasoning for long-context LLMs in mobile edge networks by dynamically matching supply and demand. Its iterative price adjustment mechanism ensures efficient resource allocation and cost optimization while accommodating diverse reasoning paths. Additionally, DDA's scalability and real-time adaptation capabilities allow it to handle fluctuating task complexities and participant heterogeneity, ensuring fairness and transparency in resource management. The auction process involves four key steps: clock broadcast, clock acceptance, clock adjustment, and market clearing, each carefully designed to ensure efficient resource distribution and optimal pricing.

1. \textbf{Clock Broadcast:}  In the clock broadcast phase, two synchronized clocks, i.e., the buyer clock $C^t_B$ and the seller clock $C^t_S$, are leveraged to set the pricing dynamics. 
   \begin{itemize}
       \item \textbf{Buyer Clock ($C^t_B$):} Starts at a high value, $p_{\text{max}}$, and decreases by a step size $\Theta^t$.
       \item \textbf{Seller Clock ($C^t_S$):} Starts at a low value, $p_{\text{min}}$, and increases by $\Theta^t$.
   \end{itemize}
   This iterative adjustment establishes a range for potential transaction prices, progressively inviting participation from both buyers and sellers. By optimizing resource allocation and pricing in real-time, the DDA minimizes delays in serving long-context LLMs. This is especially beneficial for applications requiring rapid multi-round interactions with users, such as conversational agents or real-time decision-making systems.

2. \textbf{Clock Acceptance:}
   The clock acceptance phase determines participant inclusion in the long-context LLM serving market. Buyers submit bids $k^t_i$ and join the winning set $\mathcal{M}^t_B$ if their bids meet or exceed the current buyer clock price $k^t_i \geq C^t_B$. Similarly, sellers submit offers $o^t_j$ and are added to the winning set $\mathcal{N}^t_S$ if their offers are less than or equal to the seller clock price $o^t_j \leq C^t_S$. This step ensures a dynamic and adaptive matching of supply and demand, reflecting participants' valuations in real time.

3. \textbf{Clock Adjustment:}
   If no new bids are received during a clock cycle, the auctioneer initiates the clock adjustment phase by further narrowing the price gap between buyer and seller clocks: $C^{t+1}_B = C^t_B - \Theta^t$ and $C^{t+1}_S = C^t_S + \Theta^t$. This systematic adjustment accelerates the auction process towards convergence, ensuring efficient market operation. The scalability of the DDA mechanism enables it to support a wide range of LLM tasks, from simpler single-path reasoning to more complex multi-path reasoning scenarios. By dynamically adjusting clock prices, the auction accommodates the needs of heterogeneous tasks and resource constraints.

4. \textbf{Market Clearing:}
   Market clearing occurs when the buyer and seller clocks intersect ($C^t_B < C^t_S$), signifying an equilibrium price acceptable to both parties. The final market-clearing price ($p^*$) is calculated as a weighted average of the final clock values: $p^* = \kappa C^{T-1}_B + (1-\kappa) C^{T-1}_S$, where $\kappa \in [0,1]$ determines the relative contribution of buyer and seller valuations. This ensures a fair and balanced transaction price that reflects market dynamics.

The utilities for buyers and sellers are derived from the clearing price. The utility $u_i$ of buyer $i$ is the surplus gained from transacting below their maximum willingness to pay: $u_i = k^t_i - p^*$. Conversely, a seller's utility ($\tilde{u}_j$) represents their profit from selling above their minimum acceptable price: $\tilde{u}_j = p^* - o^t_j$. Together, these individual utilities contribute to the total social welfare $SW$, which can be calculated as
\begin{equation}
SW = \sum_{i \in \mathcal{M}^t_B} u_i + \sum_{j \in \mathcal{N}^t_S} \tilde{u}_j.
\end{equation}
The social welfare captures the overall economic benefit generated by the auction. The DDA maximizes the overall economic benefit of the system by aligning the incentives of buyers and sellers with optimal resource allocation. The social welfare metric highlights the efficiency and effectiveness of the auction in generating net positive outcomes for all participants.

\subsection{Property Analysis}

The DDA is a robust mechanism designed to ensure economic sustainability, fairness, and efficiency in the dynamic and resource-constrained long-context LLM serving market such as those supporting long-context LLMs at mobile edge networks. The economic robustness of DDA is characterized by three critical properties, including individual rationality, incentive compatibility, and budget balance.

\textbf{Individual Rationality:} Individual rationality ensures that all participants, whether buyers or sellers, achieve a non-negative utility by participating in the auction. For buyers, the utility $u_i = k^t_i - p^*$ must satisfy $u_i \geq 0$, meaning that the clearing price $p^*$ does not exceed their bid $k^t_i$. Similarly, for sellers, the utility $\tilde{u}_j = p^* - o^t_j$ must satisfy $\tilde{u}_j \geq 0$, ensuring that sellers transact only if the price exceeds their minimum acceptable bid $o^t_j$. This property guarantees voluntary participation from all market actors, promoting market stability and trust.

\textbf{Incentive Compatibility:} The DDA is incentive compatible, meaning participants are encouraged to bid truthfully. For buyers, submitting a bid lower than their true valuation $k^t_i$ risks exclusion from the winning set $M^t_B$, while overstating their bid does not provide additional utility since the market-clearing price $p^*$ is determined independently of individual bids. For sellers, underbidding below $o^t_j$ risks accepting unfavorable prices, and overbidding may lead to exclusion from the winning set $N^t_S$. This mechanism ensures that strategic misreporting of bids does not result in higher utility, maintaining fairness and transparency.

\textbf{Budget Balance:} The DDA achieves strong budget balance, meaning the total payments collected from buyers equal the total payouts to sellers, ensuring no surplus or deficit for the auctioneer. This is achieved by ensuring that the number of winning buyers equals the number of winning sellers at market clearing, and all transactions occur at the uniform market-clearing price $p^*$. Thus, the auction mechanism is self-sustaining without requiring external subsidies or incurring deficits.

\textbf{Monotonicity and Criticality:} The DDA’s robustness is further reinforced through its bid monotonicity and criticality properties. Bid monotonicity ensures that winning buyers remain in the winning set $M^t_B$ by increasing their bids, and winning sellers remain in the set $N^t_S$ by lowering their bids. Criticality defines a threshold bid below or above which participants lose their winning status, ensuring that the market-clearing price is determined equitably and predictably.

The economic robustness of DDA directly translates to its effectiveness in serving long-context LLMs at mobile edge networks. By ensuring fair and stable participation through individual rationality and incentive compatibility, the mechanism aligns participants’ incentives with optimal resource allocation. Strong budget balance ensures that the system remains economically viable while dynamically adapting to the varying demands and constraints of LLM inference tasks. The monotonicity and criticality properties further enhance predictability, enabling robust market performance under dynamic conditions.

The DDA demonstrates computational efficiency in the long-context LLM serving market by leveraging a linear complexity model that scales with the number of participants and iterations. Each iteration of the auction involves constant-time clock updates ($O(1)$) and participant evaluations ($O(n + m)$), where $n$ and $m$ are the numbers of buyers and sellers, respectively. With $T$ iterations required for the buyer and seller clocks to converge, the total complexity is $O(T \cdot (n + m))$. The convergence rate is influenced by the step size $\Theta^t$, which balances precision and iteration count. This complexity ensures scalability, allowing the DDA to handle dynamic and large-scale market scenarios while maintaining efficient resource allocation and fairness in long-context LLM serving.

\section{Experimental Results}

\begin{table*}[!]
\centering\caption{Few-shot Performance of LaMDA-137B on Various Reasoning Tasks}
\label{tab:LaMDA-137B}
\resizebox{\textwidth}{!}{%
\begin{tabular}{@{}ccccccccccccccccc@{}}
\toprule
\multirow{2}{*}{Paths} & \multicolumn{2}{c}{ASDiv} & \multicolumn{2}{c}{MultiArith} & \multicolumn{2}{c}{SVAMP} & \multicolumn{2}{c}{GSM8K} & \multicolumn{2}{c}{Commonsense QA} & \multicolumn{2}{c}{Strategy QA} & \multicolumn{2}{c}{ARC} \\ \cmidrule(l){2-15} 
 & CoT & SC-CoT & CoT & SC-CoT & CoT & SC-CoT & CoT & SC-CoT & CoT & SC-CoT & CoT & SC-CoT & CoT & SC-CoT \\ \cmidrule(r){1-15}
0 & 49 & 44 & 52 & 47 & 39 & 34 & 17 & 15 & 58 & 55 & 65 & 63 & 55 & 51 \\ 
5 & 49 & 53 & 52 & 63 & 39 & 43 & 17 & 21 & 58 & 61 & 65 & 66 & 55 & 57 \\
10 & 49 & 56 & 52 & 69 & 39 & 49 & 17 & 24 & 58 & 62 & 65 & 67 & 55 & 58 \\
15 & 49 & 57 & 52 & 71  & 39 & 51 & 17 & 25 & 58 & 62 & 65 & 67 & 55 & 59 \\
20 & 49 & 58 & 52 & 73 & 39 & 52 & 17 & 26 & 58 & 62 & 65 & 68 & 55 & 59 \\
% 25 & 52 & 74 & 49 & 58 & 39 & 52 & 17 & 26 & 58 & 62 & 65 & 68 & 75 & 78 & 55 & 59 \\
% 30 & 52 & 75 & 49 & 58 & 39 & 52 & 17 & 27 & 58 & 62 & 65 & 68 & 75 & 78 & 55 & 59 \\
% 35 & 52 & 76 & 49 & 58 & 39 & 52 & 17 & 27 & 58 & 62 & 65 & 68 & 75 & 78 & 55 & 60 \\
% 40 & 52 & 77 & 49 & 58 & 39 & 53 & 17 & 28 & 58 & 62 & 65 & 68 & 75 & 78 & 55 & 60 \\ 
\bottomrule
\end{tabular}%
}
\end{table*}

% \usepackage{booktabs}
% \usepackage{graphicx}
\begin{table*}[!]
\caption{Few-shot Performance of PaLM-540B on Various Reasoning Tasks}
\label{tab:PaLM-540B}
\resizebox{\textwidth}{!}{%
\begin{tabular}{@{}ccccccccccccccccccc@{}}
\toprule
\multirow{2}{*}{Paths} & \multicolumn{2}{c}{ASDiv} & \multicolumn{2}{c}{MultiArith} & \multicolumn{2}{c}{SVAMP} & \multicolumn{2}{c}{GSM8K} & \multicolumn{2}{c}{Commonsense QA} & \multicolumn{2}{c}{Strategy QA} & \multicolumn{2}{c}{ARC} \\ \cmidrule(l){2-15} 
 & CoT & SC-CoT & CoT & SC-CoT & CoT & SC-CoT & CoT & SC-CoT & CoT & SC-CoT & CoT & SC-CoT & CoT & SC-CoT \\ \cmidrule(r){1-15}
0 & 74 & 71 & 94 & 89 & 79 & 71 & 56 & 50 & 79 & 75 & 75 & 74 & 85 & 78 \\
5 & 74 & 78 & 94 & 97 & 79 & 83 & 56 & 65 & 79 & 79 & 75 & 78 & 85 & 86 \\
10 & 74 & 80 & 94 & 98 & 79 & 85 & 56 & 70 & 79 & 80 & 75 & 79 & 85 & 87 \\
15 & 74 & 81 & 94 & 98 & 79 & 85 & 56 & 71 & 79 & 80 & 75 & 80 & 85 & 88 \\
20 & 74 & 81 & 94 & 98 & 79 & 85 & 56 & 72 & 79 & 80 & 75 & 80 & 85 & 88 \\
% 25 & 92 & 94 & 74 & 81 & 94 & 98 & 79 & 85 & 56 & 72 & 79 & 81 & 75 & 80 & 95 & 96 & 85 & 88 \\
% 30 & 92 & 94 & 74 & 81 & 94 & 98 & 79 & 85 & 56 & 73 & 79 & 81 & 75 & 80 & 95 & 96 & 85 & 88 \\
% 35 & 92 & 94 & 74 & 81 & 94 & 98 & 79 & 85 & 56 & 74 & 79 & 81 & 75 & 81 & 95 & 96 & 85 & 88 \\
% 40 & 92 & 94 & 74 & 82 & 94 & 98 & 79 & 87 & 56 & 75 & 79 & 81 & 75 & 82 & 95 & 96 & 85 & 88 \\ 
\bottomrule
\end{tabular}%
}
\end{table*}
\begin{figure*}[t]
        \vspace{-0.5cm}
    \centering
    \subfigure[The system cost of the algorithms.]{\includegraphics[width=0.32\linewidth]{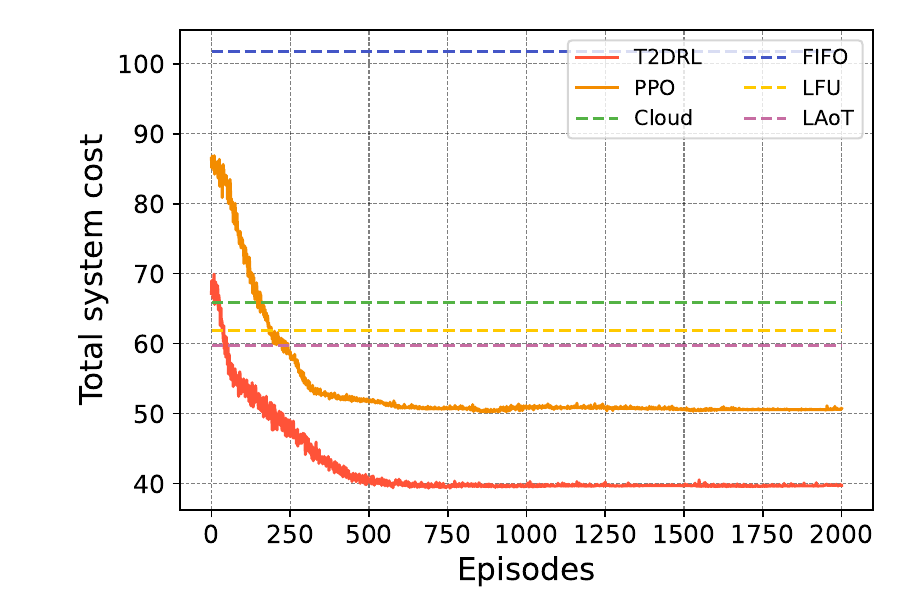}%
        \label{convergence1}}
    % \hfil
    \subfigure[The cost of T2DRL algorithm.]{\includegraphics[width=0.32\linewidth]{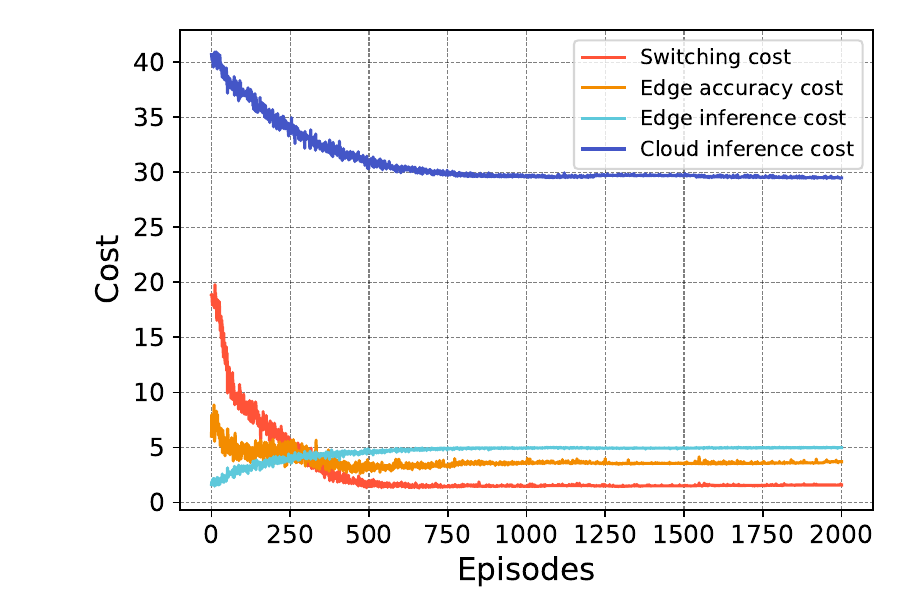}%
        \label{convergence2}}
    \subfigure[Reasoning performance of T2DRL algorithm.]{\includegraphics[width=0.32\linewidth]{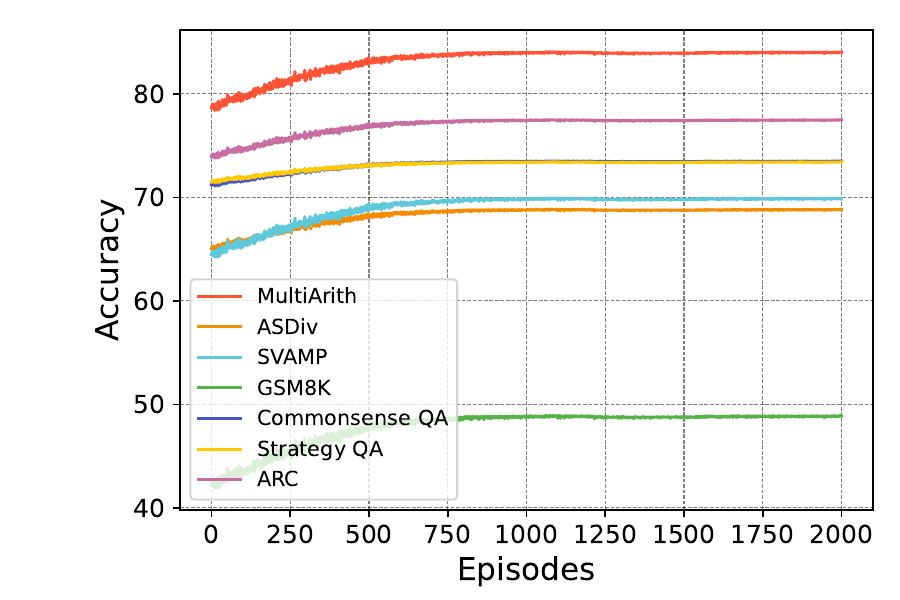}%
        \label{convergence3}}
    \caption{Convergence analysis of the proposed T2DRL algorithm.}
    \label{fig:convergence}
\end{figure*}

In the experimental setup, we consider a mobile edge system configured with varying parameters to evaluate the proposed framework and algorithm. The ESs are equipped with multiple GPUs with a total memory of 80 GB and a computing capability of 312,000 GFLOPs, supported by an energy cap of 300 W per GPU. For constructing LLM agents with perception and reasoning capability, we use the Imagebind~\cite{girdhar2023imagebind} as the perception module. In addition, LaMDA-137B and PaLM-540B~\cite{wang2022self} are adopted in the reasoning module to handle different requests. Furthermore, the network access costs are set to 0.0001 for edge access and 0.0075 for cloud access. The switching cost coefficient is set to 10$^{-5}$ and the accuracy cost coefficient is set to $2.5$. The size of input data samples is ranged from 100 to 200 tokens. The system is simulated for 100 steps to thoroughly test the performance of the proposed algorithm under dynamic network conditions.
The detailed reasoning performance of LaMDA-137B and PaLM-540B is listed in Table I and Table II.

\begin{figure*}[t]
        \vspace{-0.5cm}
    \centering
    \subfigure[Cost versus number of agents.]{\includegraphics[width=0.4\linewidth]{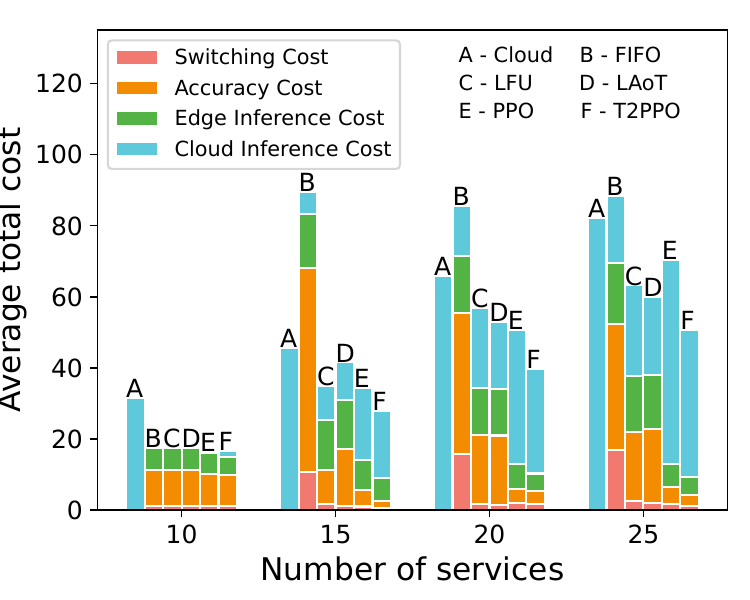}%
        \label{comparison1}}
    % \hfil
    \subfigure[Cost versus number of GPUs.]{\includegraphics[width=0.4\linewidth]{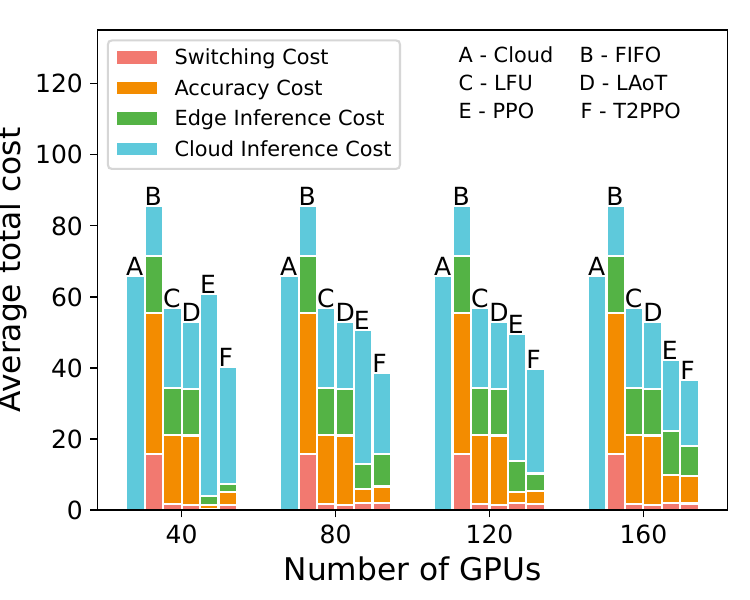}%
        \label{comparison2}}
    \subfigure[Cost versus number of reasoning paths.]{\includegraphics[width=0.4\linewidth]{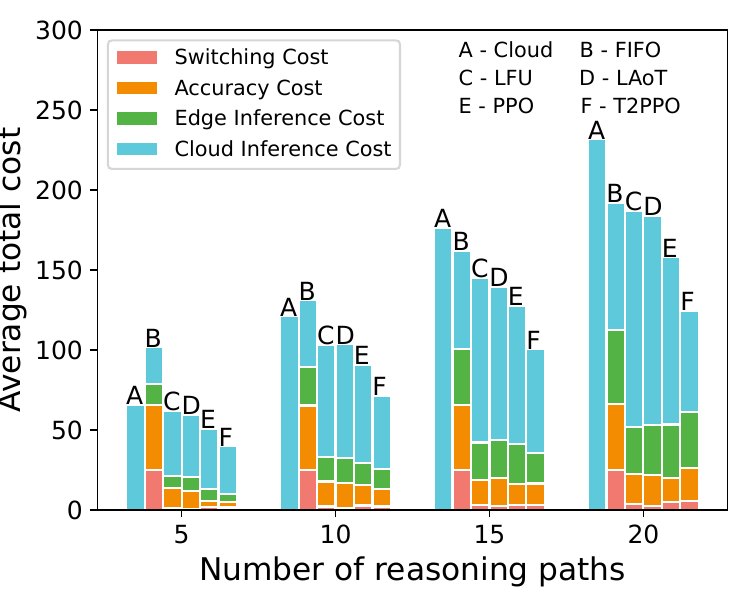}%
        \label{comparison3}}
    \subfigure[Cost versus vanishing factor.]
    {\includegraphics[width=0.4\linewidth]{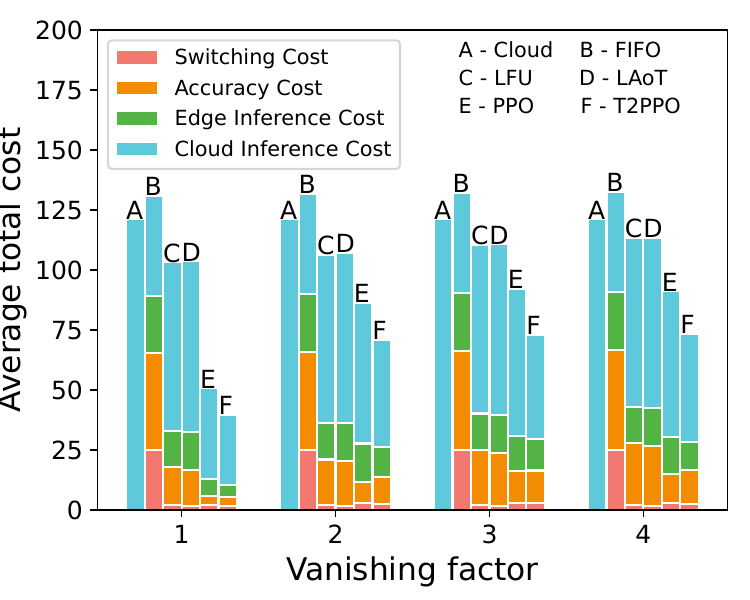}%
    \label{comparison4}}
    \caption{Performance comparison of the proposed T2DRL algorithm under different environment settings.}
    \label{fig:comparison}
\end{figure*}

In the configuration of the PPO and T2DRL, the learning rate is set to 3$\times$10$^{-4}$ and the discount factor is set to 0.95. We conducted the experiments with the seed of 42 to ensure reproducibility. We trained PPO and T2DRL over 2000 epochs, with each epoch consisting of 2000 steps. The actor-critic network of PPO utilizes hidden layers set at 128 units. The batch size is set at 128. Specific PPO adjustments included a value function coefficient of 0.25, an entropy coefficient of 0, and clipping epsilon at 0.2 for the policy updates. Additionally, the gradient is norm-capped at 0.5, with a GAE-lambda of 0.95 for advantage estimation. The TTT model is configured with a hidden size of 32, an intermediate size of 32, one hidden layer, four attention heads, and a mini-batch size of four.

% The proposed T2DRL algorithm, according to Fig.~\ref{fig:convergence}, exhibits a robust performance characterized by a rapid decrease in total system costs and a significant improvement in reasoning accuracy over 2000 episodes. According to Fig.~\ref{fig:convergence}(a), we can observe that T2DRL shows a steep decline in various costs including switching, edge accuracy, and cloud inference costs, and the T2PPO algorithm can converge to stable performance at around 500 episodes. After convergence, the T2DRL algorithm can reduce at least 30\% system cost compared with existing baselines and reduce 20\% cost compared with the PPO algorithm. Furthermore, a detailed analysis of convergence details of different costs is illustrated in Fig.~\ref{fig:convergence}(b). The decrease of total system cost is complemented by a consistent ascent in reasoning performance, as shown in Fig. \ref{fig:convergence}(c), achieving the highest accuracy by the end of the testing phase. Therefore, the convergence analysis and cost analysis demonstrate T2DRL's capability to significantly reduce system cost while enhancing the complex task-handling performance of LLM agents.

The proposed T2DRL algorithm demonstrates robust performance in reducing system costs and enhancing reasoning accuracy over 2000 episodes, as illustrated in Fig. 3. In Fig. 3(a), the T2DRL algorithm shows a rapid decline in total system costs during the initial episodes, achieving convergence at approximately 500 episodes. Compared to baseline algorithms such as FIFO, LFU, and Cloud, T2DRL achieves at least a 30\% reduction in system costs, while outperforming the PPO algorithm by 20\% in cost efficiency after convergence. This significant reduction highlights the algorithm's effectiveness in dynamic resource allocation and context-aware inference optimization. Fig. 3(b) provides a breakdown of cost components, including switching costs, edge accuracy costs, edge inference costs, and cloud inference costs. The results reveal that T2DRL effectively minimizes switching costs and balances the trade-offs between edge and cloud inference costs, stabilizing these values as the episodes progress. Notably, cloud inference costs experience a sharp decline early on, indicating the algorithm's ability to prioritize edge-based solutions for improved cost efficiency. The edge accuracy cost remains stable after convergence, reflecting the algorithm’s ability to maintain high inference quality while managing resources efficiently. Furthermore, Fig. 3(c) highlights the reasoning performance of the T2DRL algorithm across various datasets, such as MultiArith, ASDiv, SVAMP, GSM8K, and Commonsense QA. The reasoning accuracy consistently improves throughout the training process, reaching a peak by the final episodes. T2DRL achieves the highest accuracy across all tested tasks, demonstrating its capability to handle complex reasoning tasks effectively. This improvement underscores the algorithm's ability to adaptively leverage context windows and optimize long-context LLM serving. In conclusion, the convergence analysis, as detailed in Fig. 3, underscores the T2DRL algorithm's ability to significantly reduce system costs while consistently enhancing the reasoning performance of LLM agents, making it a highly effective framework for dynamic and resource-constrained mobile edge environments.

% Figure~\ref{fig:comparison} demonstrates the superior performance of the T2DRL algorithm compared to various baselines, including, Cloud, FIFO, LFU, Least Age-of-Thought (LAoT)~\cite{xu2024cached}, and PPO, across different system settings. In terms of the number of agents, as illustrated in Fig.~\ref{fig:comparison}(a), T2DRL consistently maintains the lowest total system cost as the number of agents increases from 10 to 25, outperforming other algorithms by effectively managing multiple agents and reducing switching, accuracy, edge inference, and cloud inference costs. When considering the number of GPUs in Fig.~\ref{fig:comparison}(b), T2DRL shows a significant reduction in total costs as the number of GPUs increases from 40 to 160, demonstrating its ability to efficiently utilize additional computational resources, resulting in lower costs than those of existing baselines. Evaluating the number of reasoning paths in Fig.~\ref{fig:comparison}(c), T2DRL achieves substantial cost savings by optimizing the management of reasoning paths, with total costs decreasing as the number of paths increases from 5 to 20, highlighting its effectiveness in minimizing switching and accuracy costs. Lastly, under varying vanishing factors Fig.~\ref{fig:comparison}(d), T2DRL consistently exhibits the lowest average total cost under the number of reasoning paths equal to 10, showcasing its robustness in handling scenarios where context relevance diminishes over time. 
\begin{figure*}[t]
    \centering
    \includegraphics[width=0.8\linewidth]{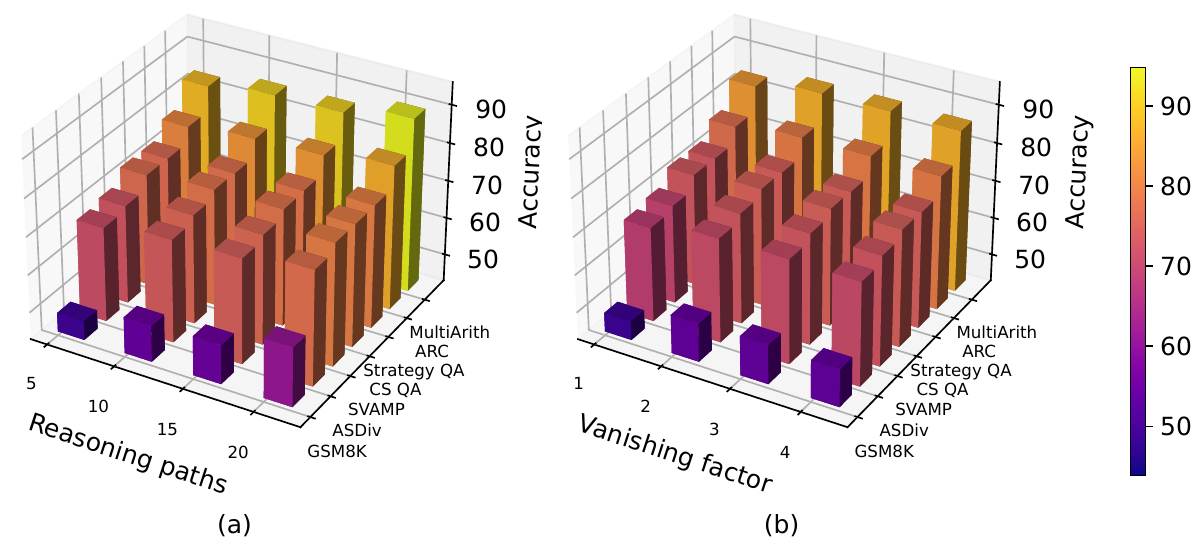}
    \caption{Reasoning accuracy of the T2DRL algorithm under different (a) reasoning paths; (b) vanishing factor.}
    \label{fig:comparison6}
\end{figure*}
Figure 4 showcases the superior performance of the T2DRL algorithm compared to various baseline methods, including Cloud, FIFO, LFU, Least Age-of-Thought (LAoT), and PPO, across diverse system configurations. As shown in Fig. 4(a), T2DRL consistently achieves the lowest total system cost as the number of agents increases from 10 to 25. This significant cost reduction is attributed to T2DRL's effective management of multiple agents, which minimizes switching costs while optimizing edge and cloud inference costs. In contrast, baseline algorithms such as FIFO and Cloud exhibit substantial increases in costs as the number of agents grows, highlighting T2DRL's scalability and resource efficiency in high-demand scenarios. Fig. 4(b) demonstrates that T2DRL effectively reduces total costs as the number of GPUs increases from 40 to 160. This result underscores T2DRL's ability to adapt to additional computational resources by efficiently distributing workloads and leveraging GPU capacity. Compared to baseline methods, T2DRL achieves consistently lower costs across all GPU configurations, with significant savings observed in edge inference costs and reduced reliance on cloud inference. When analyzing the number of reasoning paths, as depicted in Fig. 4(c), T2DRL shows remarkable cost savings as the number of paths increases from 5 to 20. The algorithm optimally manages reasoning paths by reducing switching and accuracy costs, which dominate in scenarios requiring complex multi-path reasoning. While the baseline algorithms struggle to maintain efficiency with an increasing number of paths, T2DRL demonstrates resilience and maintains low total system costs. Lastly, Fig. 4(d) evaluates the performance of T2DRL under varying vanishing factors, which represent the diminishing relevance of contextual information over time. T2DRL exhibits the lowest average total cost across all vanishing factor settings, effectively mitigating the impact of context loss by balancing edge inference and switching costs. In comparison, baseline algorithms like Cloud and FIFO display higher costs due to inefficient context management and over-reliance on cloud resources. Overall, the experimental results in Fig. 4 highlight T2DRL's robust adaptability and efficiency across different environmental settings. By minimizing switching, accuracy, edge inference, and cloud inference costs, T2DRL significantly outperforms existing baselines in scenarios with varying agent numbers, GPU availability, reasoning paths, and vanishing factors. These findings underscore T2DRL's potential as a highly effective solution for dynamic and resource-constrained mobile edge networks.
  
% Specifically, Fig.~\ref{fig:comparison6} illustrates the reasoning performance of the T2DRL algorithm tested on real-world datasets, including MultiArith, ARC, Strategy QA, Commonsense QA, SVAMP, ASDiv, and GSM8K, showing that as the number of reasoning paths increases from 5 to 20 in Fig.~\ref{fig:comparison6}(a), and the vanishing factor rises from 1 to 4 in Fig.~\ref{fig:comparison6}(b), the accuracy consistently improves across various datasets. A higher value is achieved in conditions with more reasoning paths and higher vanishing factors. This indicates that the T2DRL algorithm effectively utilizes additional reasoning paths and maintains context relevance more effectively with higher vanishing factors, leading to enhanced reasoning accuracy across different datasets.
Specifically, Fig.~\ref{fig:comparison6} illustrates the reasoning performance of the T2DRL algorithm tested on real-world datasets, including MultiArith~\cite{roy2016solving}, ARC~\cite{clark2018think}, Strategy QA~\cite{geva2021did}, Commonsense QA~\cite{talmor2018commonsenseqa}, SVAMP~\cite{wang2021adversarial}, ASDiv~\cite{zhang2019improve}, and GSM8K~\cite{cobbe2021training}, showing that as the number of reasoning paths increases from 5 to 20 in Fig.~\ref{fig:comparison6}(a), and the vanishing factor rises from 1 to 4 in Fig.~\ref{fig:comparison6}(b), the accuracy consistently improves across various datasets. As illustrated in Fig. 5(a), the reasoning accuracy consistently improves as the number of reasoning paths increases from 5 to 20 across all datasets. The datasets with simpler reasoning requirements, such as GSM8K and SVAMP, show moderate improvements, while datasets like MultiArith and ARC, which demand more complex reasoning, exhibit steeper accuracy gains with additional paths. This trend highlights T2DRL's ability to leverage multi-path reasoning effectively, balancing computational overhead with improved inference quality. The results indicate that T2DRL dynamically optimizes resource allocation for reasoning paths, enabling the LLM agents to explore diverse reasoning chains and enhance overall performance. In Fig. 5(b), reasoning accuracy improves with an increase in the vanishing factor, which reflects the relevance of context over time. Datasets that rely heavily on maintaining contextual consistency, such as Strategy QA and Commonsense QA, see the most significant improvements. As the vanishing factor rises from 1 to 4, T2DRL demonstrates a remarkable ability to maintain and utilize context effectively, preventing degradation in inference quality. This result underscores the algorithm's robustness in scenarios where context relevance decays over time, as it adapts to mitigate information loss and prioritize critical reasoning components. Across both subfigures, T2DRL achieves high accuracy consistently, with MultiArith and ARC emerging as the top-performing datasets, achieving accuracy levels close to 90\% in favorable conditions. This pattern indicates that T2DRL is well-suited for handling a wide variety of reasoning tasks, from numerical and logical reasoning (MultiArith) to commonsense and contextual reasoning (Commonsense QA, Strategy QA).
% An example of a floating figure using the graphicx package.
% Note that \label must occur AFTER (or within) \caption.
% For figures, \caption should occur after the \includegraphics.
% Note that IEEEtran v1.7 and later has special internal code that
% is designed to preserve the operation of \label within \caption
% even when the captionsoff option is in effect. However, because
% of issues like this, it may be the safest practice to put all your
% \label just after \caption rather than within \caption{}.
%
% Reminder: the "draftcls" or "draftclsnofoot", not "draft", class
% option should be used if it is desired that the figures are to be
% displayed while in draft mode.
%
%\begin{figure}[!t]
%\centering
%\includegraphics[width=2.5in]{myfigure}
% where an .eps filename suffix will be assumed under latex, 
% and a .pdf suffix will be assumed for pdflatex; or what has been declared
% via \DeclareGraphicsExtensions.
%\caption{Simulation results for the network.}
%\label{fig_sim}
%\end{figure}

\begin{figure}
    \centering
    \includegraphics[width=1\linewidth]{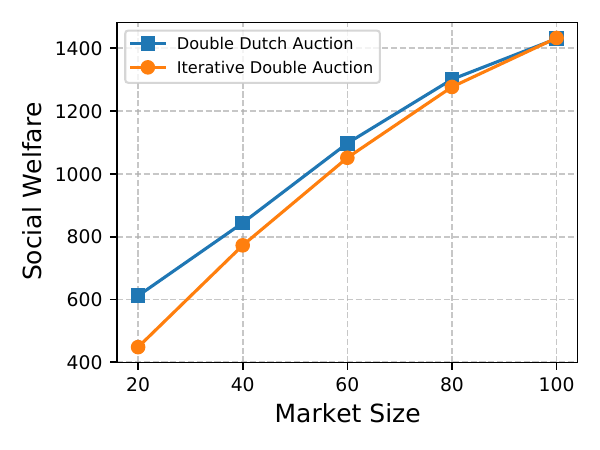}
    \caption{Social welfare verse market size.}
    \label{fig:auction}
\end{figure}

To evaluate the long-context LLM serving market and proposed auction mechanism, Fig.~\ref{fig:auction} demonstrates the social welfare performance of Iterative Double Auction (IDA) and DDA across varying market sizes. Both mechanisms exhibit a consistent upward trend in social welfare as market size increases, reflecting their ability to allocate resources more efficiently with greater market activity. However, the DDA consistently outperforms the Iterative Double Auction across all market sizes, showcasing its superior efficiency and scalability. At smaller market sizes (e.g., 20 and 40), the DDA achieves the social welfare approximately 36\% higher than IDA, indicating its ability to optimize outcomes even in resource-constrained or limited-participant scenarios. As the market size grows to 60, 80, and eventually 100, the gap between the two mechanisms narrows slightly. Despite this convergence, the DDA maintains a clear advantage, achieving higher social welfare values across the board. For instance, at a market size of 100, the DDA outperforms the IDA by approximately 12\%, demonstrating its continued robustness and adaptability in larger, more complex markets. The superior performance of the DDA can be attributed to its dynamic and synchronized pricing mechanism, which effectively matches supply and demand while maximizing participant utility. This ensures optimal resource allocation and minimizes inefficiencies, particularly in scenarios with heterogeneous market participants and varying resource demands.

\section{Conclusion}

In this paper, we have proposed a universal long-context LLM serving framework for supporting the perception and complex reasoning of LLM agents in mobile edge networks, which leverages the T2DRL algorithm to optimize model caching and inference offloading adaptively. Specifically, we have designed the system model to optimize the deployment and execution of long-context LLMs in mobile edge networks by dynamically managing model caching and inference offloading, considering resource constraints such as GPU memory and computational power. The theoretical analysis introduces bounds on ambiguity and error, assessing the system's capability to maintain accuracy and reliability under varied conditions, particularly in handling long-context demands efficiently. The proposed T2DRL algorithm can learn and adapt the optimal policy in real-time, optimizing caching and inference offloading proactively to reduce system costs and ensure the perception and reasoning performance of LLM agents in dynamic edge systems. Additionally, we have leveraged the DDA mechanism for resource allocation by dynamically matching supply and demand, ensuring fairness and transparency while maximizing the social welfare. Experimental results have highlighted the effectiveness of the proposed algorithm in reducing total system cost for LLM serving in mobile edge networks.


% \section*{Acknowledgment}

% The authors would like to thank\ldots



%\bibliographystyle{IEEEtran}


% \end{thebibliography}

\bibliographystyle{IEEEtran}
\bibliography{IEEEabrv,main}

\end{document}